\documentclass[twocolumn, aps, pra,prl, superscriptaddress, preprintnumbers, letterpaper]{revtex4-1}

\usepackage{amsthm}
\usepackage{amsmath,bm}
\usepackage{amssymb}
\usepackage{braket}
\usepackage{bbold}
\usepackage{graphicx}
\usepackage{subfigure}
\usepackage{hyperref}
\usepackage{appendix}
\usepackage{tikz}
\usepackage{color}
\usepgflibrary{shapes.geometric}
\usetikzlibrary{arrows}

\hypersetup{colorlinks=true, citecolor=blue, urlcolor=blue, linkcolor=blue}

\newcommand{\nc}{\newcommand}
\nc{\tr}{{\operatorname{Tr}}}
\newtheorem{theorem}{Theorem}

\newtheorem{definition}[theorem]{Definition}

\newtheorem{proposition}[theorem]{Proposition}
\usepackage{natbib}
\begin{document}
\title{Multipartite entanglement measures based on the thermodynamic framework}
\author{Chen-Ming Bai}
\email{baichm@stdu.edu.cn}
\affiliation{Department of Mathematics and Physics, Shijiazhuang Tiedao University,
Shijiazhuang, 050043, China}
\author{Yu Luo}
\email{luoyu@snnu.edu.cn}
\affiliation{College of Computer Science, Shaanxi Normal University, Xi'an, 710062, China}

\bigskip

\begin{abstract}
In this work, we introduce a unified method to characterize and measure multipartite entanglement using the framework of thermodynamics. A family of the new entanglement measures is proposed: \textit{ergotropic-gap concentratable entanglement}. Furthermore, we establish that ergotropic-gap concentratable entanglement constitutes a well-defined entanglement measure within a specific parameter regime, satisfying key properties including continuity, majorization monotonicity and monogamy. We demonstrate the utility of this measure by showing it effectively distinguishes between multi-qubit Greenberger-Horne-Zeilinger  states and W states. It also proves effective in detecting entanglement in specific classes of four-partite star quantum network states.
\end{abstract}


\maketitle

\section{Introduction}
Quantum entanglement, an indispensable resource, offers key technologies like quantum cryptography \cite{yin2020entanglement,gisin2002quantum,zeng2023controlled,portmann2022security}, quantum teleportation \cite{pirandola2015advances,karlsson1998quantum,hu2023progress}, and quantum secret sharing \cite{hillery1999quantum,singh2024controlled}. Moreover, advances in quantum control technologies have expanded the utility of entanglement to the demonstration of a quantum advantage in near-term devices \cite{king2025beyond}, with its performance often critically dependent on the degree of entanglement within the state. While bipartite entanglement is well-characterized by measures including concurrence \cite{hill1997entanglement,xuan2025new,zhou2025parameterized}, entanglement of formation \cite{wootters1998entanglement,wootters2001entanglement}, and crucially, thermodynamically-inspired entropy-based measures like R{\'e}nyi-$\alpha$ \cite{wang2016entanglement,liang2020monogamy}, Tsallis-$q$  \cite{kim2010tsallis,luo2016general} and Unified-$(q,s)$ entropy entanglement \cite{san2011unified,li2024monogamy}, multipartite entanglement remains less comprehensively understood. This gap arises from rapidly increasing structural complexity as subsystems grow, exemplified by the fundamental stochastic local operations and classical communication (LOCC) \cite{dur2000three}---inequivalence of Greenberger-Horne-Zeilinger (GHZ) \cite{greenberger1990bell}  and W states in tripartite systems.

Recently, quantum thermodynamics has emerged as a rapidly evolving discipline, offering promising avenues to bridge the fundamental gap between quantum mechanics and thermodynamics \cite{brandao2008entanglement,plenio1998teleportation,kaufman2016quantum}. Critically, a profound connection has been identified between thermodynamic quantities and quantum entanglement \cite{oppenheim2002thermodynamical,puliyil2022thermodynamic}. In particular, composite quantum systems have been extensively explored as potential work storage devices \cite{puliyil2022thermodynamic, delmonte2021characterization,mukherjee2016presence,alicki2013entanglement,alimuddin2019bound,alimuddin2020independence,vsafranek2023work}, leveraging the inherent resource of quantum entanglement. Crucially, the maximum extractable work from such composite systems is demonstrably enhanced when global operations across the entire system are permitted, compared to scenarios restricted to local operations on individual subsystems. This operational advantage underscores the thermodynamic relevance of entanglement. Moving beyond traditional entanglement quantifiers based on entropy, novel thermodynamic-inspired measures have been proposed to characterize entanglement. Puliyil et al.\cite{puliyil2022thermodynamic} introduced the ergotropic gap, while Yang et al.\cite{yang2023battery} defined the battery capacity gap.

Remarkably, these thermodynamic quantities not only witness the presence of bipartite entanglement but also serve as indicators of genuine multipartite entanglement (GME). This establishes a direct operational link between thermodynamic performance and entanglement structure.
The connection between thermodynamics and entanglement quantification deepened significantly in 2024. Yang et al.\cite{yang2024multiparticle,yang2024characterizing} established a rigorous relation between the geometric measure of entanglement and ergotropic gaps. Concurrently, Sun et al.\cite{sun2024genuine} introduced a vector-valued measure for GME, exploiting the battery capacity framework within multipartite settings. These developments highlight thermodynamics as a potent framework for deriving entanglement measures. However, the above research results lack a unified thermodynamic framework to characterize and quantify entanglement.

In this work, we introduce and characterize a novel family of multipartite entanglement measures for mixed states, denoted by $M^{(s)}_E$ or $M^{(s)}_{B}$. These measures are derived from fundamental quantum thermodynamic concepts: the ergotropic gap and the battery capacity gap. We refer to these measures as ergotropic-gap concentratable entanglement and battery capacity-gap concentratable entanglement. The concentratable entanglement framework \cite{beckey2021computable} provides a flexible tool for analyzing entanglement distribution across different partitions. Crucially, these measures establish a direct quantitative link between the operational resource of entanglement and thermodynamic work extraction capabilities in quantum systems.
Furthermore, we rigorously prove that within systems governed by Hamiltonians featuring equispaced energy levels, these two entanglement measures are equivalent. After that, we establish that $M^{(s)}_E$ satisfies essential axioms for a valid entanglement measure: it is non-increasing on average under LOCC and exhibits continuity.  Leveraging this thermodynamic perspective, we employ the
ergotropic-gap concentratable entanglement to develop a sufficient criterion for certifying GME in three-qubit systems. The criterion directly connects the impossibility of achieving certain work extraction benchmarks under LOCC to the presence of genuine multipartite correlations. We also explore the monogamy relations satisfied by $M^{(s)}_E$ and investigate its relationships with other established entanglement measures. Finally, we investigate the ability
of $M^{(s)}_E$ to distinguish between GHZ and W states. Our
results show that this measure provides a clear distinction
between these two inequivalent classes of multipartite
entangled states. Intriguingly, in specific four-partite star-shaped quantum network configurations, we find that in some special
cases the ergotropic-gap concentratable entanglement is higher than the original entanglements.

This paper is organized as follows. In Sec.~\ref{sec:Preliminary}, we introduce some preliminaries and give two definitions-ergotropic-gap concentratable entanglement and battery capacity-gap concentratable entanglement. In Sec.~\ref{sec:mainresults}, we provide our main results.  In Sec.~\ref{sec:example}, we investigate the
ability of $M^{(s)}_E$ to distinguish between GHZ and W states,
and study the example of a four-partite star quantum
network.  We summarize our results in Sec.~\ref{sec:Conclusion}.

\section{Preliminary} \label{sec:Preliminary}
\subsection{The ergotropic gap and battery capacity gap}
The study of work extraction from isolated quantum systems via cyclic Hamiltonian processes dates to the late 1970s (\cite{pusz1978passive,lenard1978thermodynamical}). Consider a finite-dimensional system initialized in the state $\rho$ on Hilbert space $\mathcal{H}$. Its time evolution is driven by external modulations of the Hamiltonian  $H=\sum_j\epsilon_j\ket{\epsilon_j}\bra{\epsilon_j}$ with nondegenerate eigenvalues ordered such that $\epsilon_j\leq\epsilon_{j+1}$ for $j=0,1,\cdots,d-1$, and where the ground-state energy satisfies $\epsilon_0=0$ with $\epsilon_1>0$. The maximal extractable work from the initial state $\rho$, known as \textit{ergotropy} \cite{allahverdyan2004maximal}, is then defined as
\begin{equation}
    W_e(\rho)=\tr(\rho H)-\tr(\rho^p H),
\end{equation}
where \(\rho^p\) is the corresponding passive state of \(\rho\) and is given by $\rho^p=\sum_j\lambda_j\ket{\epsilon_j}\bra{\epsilon_j}$ with $\lambda_j\geq\lambda_{j+1}$.

Similarly, we define the \textit{anti-ergotropy} \cite{yang2023battery} as the minimal extractable work from the system through cyclic processes, expressed as
\begin{equation}
    W_{ae}(\rho)=\tr(\rho H)-\tr(\rho^{ac} H),
\end{equation}
where \(\rho^{ac}\) denotes the active state of \(\rho\) that maximizes the mean energy $\langle H\rangle$.

The \textit{battery capacity} of the system is defined as \cite{yang2023battery}
\begin{equation}
  \mathcal{C}(\rho)=W_{e}(\rho)-W_{ae}(\rho)=\tr(\rho^{ac} H)-\tr(\rho^{p} H).
\end{equation}
This capacity quantifies the difference between the maximum and minimum achievable mean energies under unitary evolution, equaling the difference between the ergotropy and anti-ergotropy of \(\rho\).

Furthermore, for a bipartite state $\rho_{AB}\in\mathcal{D}(\mathcal{H}_A\otimes\mathcal{H}_B)$, the global Hamiltonian is $H_{AB}=H_A\otimes \mathbb{I}_B+\mathbb{I}_A\otimes H_B$, where $\mathbb{I}_X$ is the identity operator on $\mathcal{H}_X$.  Each local Hamiltonian has spectral decomposition
$H_X=\sum_j\epsilon_j^{X}\ket{\epsilon_j^{X}}\bra{\epsilon_j^{X}}\ (X=A,B)$ with eigenvalues ordered as $\epsilon^X_j\leq\epsilon^X_{j+1}$.
The global ergotropy (maximal work extractable via joint operations) is defined as
\begin{equation}
    W_e^g(\rho_{AB})=\tr(\rho_{AB}H_{AB})-\tr(\rho^p_{AB}H_{AB}),
\end{equation}
where \(\rho^p_{AB}\) is the global passive state of \(\rho_{AB}\). Conversely, the local ergotropy (maximal work via local operations) is
\begin{eqnarray}
 W_e^l(\rho_{AB})&=&W_e^A(\rho_{AB})+W_e^B(\rho_{AB})\\
&=&\tr(\rho_{AB}H_{AB})-\tr(\rho^p_{A}H_{A})-\tr(\rho^p_{B}H_{B}),  \nonumber
\end{eqnarray}
where \(\rho^p_{X}\) is the local passive states. The \textit{ergotropic gap}, quantifying the quantum advantage of global operations, is \cite{alimuddin2019bound,puliyil2022thermodynamic,yang2024multiparticle}
\begin{eqnarray}
\Delta_{A|B}(\rho_{AB})&=&W_e^g(\rho_{AB})-W_e^l(\rho_{AB})\\
&=&\tr(\rho^p_{A}H_{A})+\tr(\rho^p_{B}H_{B})
-\tr(\rho^p_{AB}H_{AB}). \nonumber
\end{eqnarray}

For pure states \(\rho_{AB}=\ket{\psi_{AB}}\bra{\psi_{AB}}\), this simplifies to
\begin{equation}
\label{eq:Delta}
    \Delta_{A|B}(\ket{\psi_{AB}})
=\tr(\rho^p_{A}H_{A})+\tr(\rho^p_{B}H_{B}),
\end{equation}
which means that \(\Delta_{A|B}(\ket{\psi_{AB}})\) represents the total extractable energy, derived from measuring the passive states of each subsystem $A$ and $B$.

Next, the bipartite battery capacity gap \(\Delta^{cg}_{A|B}(\rho_{AB})\) is the difference between  global battery capacity $\mathcal{C}^g(\rho_{AB})$ and local battery capacity $\mathcal{C}^l(\rho_{AB})$ \cite{yang2023battery}. Therefore,  the \textit{bipartite battery capacity gap} is denoted by
\begin{eqnarray}
\label{eq:Deltacg}
   \Delta^{cg}_{A|B}(\rho_{AB})&=&\mathcal{C}^g(\rho_{AB})-\mathcal{C}^l(\rho_{AB})\nonumber\\
   &=&\mathcal{C}(\rho_{AB})-\mathcal{C}(\rho_{A})-\mathcal{C}(\rho_{B}).
\end{eqnarray}
\subsection{The ergotropic-gap concentratable
entanglement and battery capacity-gap concentratable
entanglement}

In this section,  we introduce multipartite entanglement measures. A function
 $E^{(m)}: \mathcal{S}^{A_1A_2\cdots A_m} \rightarrow \mathbb{R}_+$ is called an $m$-partite entanglement measure if it satisfies the following criteria \cite{hong2012measure,hiesmayr2008multipartite,horodecki2009quantum}:
\begin{itemize}
    \item (E1) $E^{(m)}(\rho)=0$ if $\rho$ is fully separable.
    \item (E2) $E^{(m)}$ cannot increase under $m$-partite LOCC.
\end{itemize}
In addition, $E^{(m)}$ is said to be an $m$-partite entanglement monotone if it is convex and does not increase on average under $m$-partite stochastic LOCC.

Recent studies demonstrate that the ergotropic gap and battery capacity gap serve as effective entanglement quantifiers \cite{puliyil2022thermodynamic,mukherjee2016presence,yang2023battery,yang2024multiparticle,sun2024genuine}. We now generalize these concepts to multipartite systems by introducing two families of the new entanglement measures: \textit{the ergotropic-gap concentratable entanglement measures} and \textit{the battery capacity-gap concentratable entanglement measures}.

\begin{definition}[\textit{ergotropic-gap
 concentratable entanglement}]
For an $n$-qubit pure state $\ket{\psi}$ and a subset $s \subseteq [n]$ of the qubits, the ergotropic-gap concentratable entanglement for pure states is defined as
\begin{equation}
M^{(s)}_E(|\psi\rangle) = \frac{1}{2^{|s|}} \sum_{X \in \mathcal{P}(s)}\Delta_{X|X^c}(|\psi\rangle),
\end{equation}
where $\mathcal{P}(s)$ denotes the power set of $s$, $X^c$ denotes the complement of the subsystem $X$, and $\Delta_{X|X^c}(|\psi\rangle)=0$ when $X = \emptyset$.
For mixed states \(\rho\), the measure extends via convex roof construction:
\begin{equation}
\label{eq:deff1}
M^{(s)}_E(\rho) =\min\sum_ip_iM^{(s)}_E(\ket{\psi_i}) ,
\end{equation}
where the minimum is taken over all the pure state decomposition ${\{p_i,\ket{\psi_i}\bra{\psi_i}\}}$ of $\rho$.
\end{definition}

\begin{definition}[\textit{battery capacity-gap
 concentratable entanglement}]
For an $n$-qubit pure state $\ket{\psi}$ and a subset $s\subseteq [n]$ of the qubits, the battery capacity-gap concentratable entanglement for pure states is defined as
\begin{equation}
M^{(s)}_{B}(|\psi\rangle) = \frac{1}{2^{|s|}} \sum_{X \in \mathcal{P}(s)}\Delta^{cg}_{X|X^c}(|\psi\rangle),
\end{equation}
where $\mathcal{P}(s)$ denotes the power set of $s$,  $X^c$ denotes the complement of the subsystem $X$, and $\Delta_{X|X^c}(|\psi\rangle)=0$ when $X = \emptyset$.
For any mixed state \(\rho\), the battery capacity-gap concentratable entanglement is defined as
\begin{equation}
M^{(s)}_{B}(\rho) =\min\sum_ip_iM^{(s)}_{B}(\ket{\psi_i}) ,
\end{equation}
where the minimum is taken over all the pure state decomposition ${\{p_i,\ket{\psi_i}\bra{\psi_i}\}}$ of $\rho$.
\end{definition}

Furthermore, in a Hamiltonian with equispaced energy levels, the ergotropic-gap
 concentratable entanglement and the battery capacity-gap concentratable
entanglement are equivalent:
\begin{theorem}
For a system with equispaced energy level Hamiltonians $H_X=\sum_{j=0}^{d-1}j\epsilon^{X}\ket{\epsilon_j^{X}}\bra{\epsilon_j^{X}}$ and $H_{X^c}=\sum_{j=0}^{d-1}j\epsilon^{X^c}\ket{\epsilon_j^{X^c}}\bra{\epsilon_j^{X^c}}$, the battery capacity-gap and ergotropic-gap concentratable entanglement measures satisfy
\begin{equation}
 M^{(s)}_{B}(|\psi\rangle)=2M^{(s)}_E(|\psi\rangle),
\end{equation}
for any $n$-qubit pure state $\ket{\psi}$ and  qubit subset $s\subseteq [n]$.
\end{theorem}
\begin{proof}
The passive and active states of subsystems
$X$ and $X^c$ admit spectral decompositions:
\begin{eqnarray}
    &&\rho_X^p=\sum_{j=0}^{d-1}\lambda^X_j\ket{\epsilon_j^{X}}\bra{\epsilon_j^{X}}, \rho_{X^c}^p=\sum_{j=0}^{d-1}\lambda^{X^c}_j\ket{\epsilon_j^{X^c}}\bra{\epsilon_j^{X^c}},\nonumber\\
    &&\rho_X^{ac}=\sum_{j=0}^{d-1}\eta^X_j\ket{\epsilon_j^{X}}\bra{\epsilon_j^{X}},\rho_{X^c}^{ac}=\sum_{j=0}^{d-1}\eta^{X^c}_j\ket{\epsilon_j^{X^c}}\bra{\epsilon_j^{X^c}},\nonumber
\end{eqnarray}
with eigenvalue relations  $\lambda^X_j=\eta^X_{d-1-j} $ and $ \lambda^{X^c}_j=\eta^{X^c}_{d-1-j}$. From these equalities above, we deduce that
\begin{eqnarray}
\label{eq:rhox}
    (d-1)\epsilon^{X}-\tr(\rho^{ac}_XH_X)&=&\sum_{j=0}^{d-1}\eta_j^{X}((d-1)\epsilon^{X}-j\epsilon^{X})\nonumber\\
    &=& \sum_{j=0}^{d-1}\lambda_{d-1-j}^{X}(d-1-j)\epsilon^{X}\nonumber\\
    &=& \tr(\rho_{X}^pH_X).
\end{eqnarray}
Using the same method, we can get
\begin{equation}
\label{eq:rhoxc}
   (d-1)\epsilon^{X^c}-\tr(\rho^{ac}_{X^c}H_{X^c})
    = \tr(\rho_{X^c}^pH_{X^c}).
\end{equation}
Adding Eq.(\ref{eq:rhox}) and Eq.(\ref{eq:rhoxc}), we can obtain that
\begin{eqnarray}
    \label{eq:rhoadd}
   && (d-1)\epsilon^{X}-\tr(\rho^{ac}_XH_X)+(d-1)\epsilon^{X^c}-\tr(\rho^{ac}_{X^c}H_{X^c})\nonumber\\
    &&=\tr(\rho_{X}^pH_X)+\tr(\rho_{X^c}^pH_{X^c}).
\end{eqnarray}
Furthermore, we can have that
\begin{eqnarray}
    \label{eq:rhoadd}
   && \quad(d-1)\epsilon^{X}-\tr(\rho^{ac}_XH_X)+\tr(\rho_{X}^pH_X)\nonumber\\
   &&+(d-1)\epsilon^{X^c}-\tr(\rho^{ac}_{X^c}H_{X^c})+\tr(\rho_{X^c}^pH_{X^c})\nonumber\\
    &&=2(\tr(\rho_{X}^pH_X)+\tr(\rho_{X^c}^pH_{X^c})).
\end{eqnarray}

By Eq.(\ref{eq:Delta}) and Eq.(\ref{eq:Deltacg}),  Eq.(\ref{eq:rhoadd}) is changed to
$\Delta^{cg}_{X|X^c}(|\psi\rangle)=2\Delta_{X|X^c}(|\psi\rangle)$. Thus, by averaging over all bipartitions $X \in \mathcal{P}(s)$, we can get that $M^{(s)}_{B}(|\psi\rangle)=2M^{(s)}_E(|\psi\rangle)$.
\end{proof}
After that, we establish that the ergotropic-gap concentratable entanglement $ M^{(s)}_E$ constitutes a valid multipartite entanglement measure by proving its adherence to two fundamental criteria, i.e.,
non-negative and non-increasing on average under multipartite LOCC protocols.
While analogous properties hold for the battery capacity-gap variant, we focus here on $ M^{(s)}_E$ for clarity.

\begin{theorem}\label{the:vNCE}
For any $n$-qubit state $\rho$ and given the subset $s\subseteq [n]$, the ergotropic-gap concentratable entanglement $M^{(s)}_E(\rho)$ is a multipartite entanglement measure under multipartite LOCC protocol,
i.e., $M^{(s)}_E(\rho)$ satisfies following properties:

(1)~\textbf{Non-negativity}: $M^{(s)}_E(\rho)\geq0$. Moreover, the equality holds if and only if
$\rho$ is fully separable, i.e., $\rho=\sum_ip_i\bigotimes_{j=1}^n\ket{\psi_{ij}}\bra{\psi_{ij}}$, where each $\ket{\psi_{ij}}$ is a single-qubit pure state.

(2)~\textbf{Monotonicity}: $M^{(s)}_E(\rho)$ is non-increasing on average under the multipartite LOCC protocols.
\end{theorem}
\begin{proof}
Firstly, we establish the measure properties for pure states. Then, we extend to mixed states via convex roof extension.
Therefore, for any fully separable pure state $\ket{\psi}=\bigotimes_{j=1}^n \ket{\psi_j}$, $M^{(s)}_E(\ket{\psi})$ is a multipartite entanglement measure.

(1) Since $\ket{\psi}$ is a fully separable state, then the state $\ket{\psi}$ must be a product pure state in an arbitrary bipartition, i.e., $\ket{\psi}=\ket{\psi}_{X}\otimes\ket{\psi}_{X^c}$. Therefore, by setting $\rho=\ket{\psi}\bra{\psi}$, both $\rho_X^p$ and $\rho_{X^c}^p$ will be
the ground states of their local Hamiltonians. Hence, we have that
\begin{equation}
    \label{eq:non1}
    \Delta_{X|X^c}(\ket{\psi})=\tr(\rho^p_{X}H_{X})+\tr(\rho^p_{X^c}H_{X^c})=0.
\end{equation}
By Eq.(\ref{eq:non1}), it implies that $M^{(s)}_E(\ket{\psi})=0$. Furthermore, for mixed separable states $\rho=\sum_ip_i\bigotimes_{j=1}^n\ket{\psi_{ij}}\bra{\psi_{ij}}$, where each $\ket{\psi_{ij}}$ is a single-qubit pure state, the convex roof construction ensures \( M^{(s)}_E(\rho)=0 \).

(2) We first show that $M^{(s)}_E(\ket{\psi})$ is non-increasing on average under the pure state to pure state transformation. In Ref.\cite{puliyil2022thermodynamic}, the $k$-separable ergotropic gap $\Delta_{X_1|\cdots|X_k}$ of a pure state $\ket{\psi}_{A_1\cdots A_n}$ is non-increasing under LOCC $\Lambda$, i.e., $\Delta_{X_1|\cdots|X_k}(\ket{\psi})\geq \Delta_{X_1|\cdots|X_k}(\Lambda(\ket{\psi}))$ with $2\leq k\leq n$. Therefore, $\Delta_{X|X^c}(\ket{\psi})\geq \Delta_{X|X^c}(\Lambda(\ket{\psi}))$ holds. Furthermore, we have that
\begin{eqnarray}
   \label{eq:non2}
   M^{(s)}_E(\ket{\psi})&=&\frac{1}{2^{|s|}}\sum_{X\in\mathcal{P}(s)}\Delta_{X|X^c}(\ket{\psi})\nonumber\\
   &\geq& \frac{1}{2^{|s|}}\sum_{X\in\mathcal{P}(s)}\Delta_{X|X^c}(\Lambda(\ket{\psi}))\nonumber\\
   &=& M^{(s)}_E(\Lambda(\ket{\psi})).
\end{eqnarray}
Next, we consider the case of transforming a pure state into a mixed-state ensemble. Let $\Lambda = \sum_k \Lambda_k$ represent a multipartite LOCC protocol that transforms a pure state $ \ket{\psi_j}$ into a mixed-state ensemble $\{ p_{jk}, \rho_{jk} \}$, where
\begin{equation}
\label{eq:ensemble}
  \rho_{jk} = \frac{1}{p_{jk}} \Lambda_k(\ket{\psi_j})\ {\rm and}\ p_{jk} = \tr[\Lambda_k(\ket{\psi_j})].
\end{equation}

Assume that
\begin{equation}
\label{eq:rhojkl}
   \rho_{jk} = \sum_l p_{jkl} \ket{\psi_{jkl}}\bra{\psi_{jkl}},
\end{equation}
and it is the optimal decomposition for $ M^{(s)}_E(\rho_{jk})$, so it implies that
\begin{equation}
   M^{(s)}_E(\rho_{jk}) = \sum_l p_{jkl} M^{(s)}_E(\ket{\psi_{jkl}}).
\end{equation}
Under these assumptions, we have the following inequality:
\begin{eqnarray}\label{eq:pure-mixed}
    M^{(s)}_E(\ket{\psi_j})
&\geq& \sum_k \sum_l p_{jk} p_{jkl} M^{(s)}_E(\ket{\psi_{jkl}})\nonumber\\
&=& \sum_k p_{jk} M^{(s)}_E(\rho_{jk}).
\end{eqnarray}
Here, the first inequality holds follow from $M^{(s)}_E(\ket{\psi_j})$ is non-increasing, on average, under the pure state to pure state transformation.

Finally, we consider transforming a mixed state $\rho$ into a mixed-state ensemble $\{q_k, \rho_k\}$ under a multipartite LOCC protocol $\Lambda=\sum_k\Lambda_k$. Suppose $\rho=\sum_jp_j\ket{\psi_j}\bra{\psi_j}$ is the optimal decomposition of $M^{(s)}_E(\rho)$, i.e., $M^{(s)}_E(\rho)=\sum_jp_jM^{(s)}_E(\ket{\psi_j})$. Then, we have
\begin{eqnarray}
\label{eq:decomposition}
\Lambda(\rho)&=&\nonumber \sum_k\Lambda_k(\rho)
\\&=&\nonumber
\sum_k\Lambda_k(\sum_jp_j\ket{\psi_j}\bra{\psi_j})
\\&=&\nonumber
\sum_k\sum_jp_j\Lambda_k(\ket{\psi_j}\bra{\psi_j})
\\&=&
\sum_k\sum_jp_jp_{jk}\rho_{jk},
\end{eqnarray}
where the last equality in Eq.(\ref{eq:decomposition}) can be obtained from Eq.(\ref{eq:ensemble}). Due to Eq.(\ref{eq:rhojkl}) and we can see that $\rho_k=\frac{1}{q_k}\Lambda_k(\rho)$ with $q_k=\tr[\Lambda_k(\rho)]$. Hence,
by Eq.(\ref{eq:pure-mixed}), we obtain that
\begin{eqnarray}
\label{eq:mix1}
M^{(s)}_E(\rho) &=&\nonumber \sum_jp_jM^{(s)}_E(\ket{\psi_j})
\\&\geq&
\sum_jp_j\sum_kp_{jk}M^{(s)}_E(\rho_{jk}).
\end{eqnarray}
Due to $$\sum_{jl}\frac{p_jp_{jk}p_{jkl}}{q_k}M^{(s)}_E(\ket{\psi_{jkl}})\geq M^{(s)}_E(\rho_k) $$
which follows from
$\rho_k=\sum_{jl}\frac{p_jp_{jk}p_{jkl}}{q_k}\ket{\psi_{jkl}}\bra{\psi_{jkl}}$ and  $q_k=\tr\Big[\sum_{jl}p_jp_{jk}p_{jkl}\ket{\psi_{jkl}}\bra{\psi_{jkl}}\Big]$ and from Eq.(\ref{eq:deff1}),  we have that
\begin{eqnarray}
\label{eq:mix2}
&&\sum_jp_j\sum_kp_{jk}M^{(s)}_E(\rho_{jk})\nonumber\\
&=&
\sum_j\sum_k\sum_lp_jp_{jk}p_{jkl}M^{(s)}_E(\ket{\psi_{jkl}})
\nonumber \\
&{\geq}&
\sum_kq_kM^{(s)}_E(\rho_k).
\end{eqnarray}
Therefore, by Eq.(\ref{eq:mix1}) and Eq.(\ref{eq:mix2}), we can obtain that $M^{(s)}_E(\rho)\geq\sum_kq_kM^{(s)}_E(\rho_k)$.
\end{proof}

\section{Properties of ergotropic-gap concentratable entanglement}\label{sec:mainresults}
In this section, we show some properties of ergotropic-gap concentrarable entanglement measures such as contiunity, majorization monotonicity and monogamy.  To prove this majorization monotonicity, we first give the majorization criterion \cite{hiroshima2003majorization} as follows.

A state \(\rho\) is said to be \emph{majorized} by a state \(\sigma\), denoted \(\lambda(\rho) \prec \lambda(\sigma)\), if
\[
\sum_{i=1}^k p_i^\downarrow \leq \sum_{i=1}^k q_i^\downarrow \quad (1 \leq k \leq n-1)
\]
and
\[
\sum_{i=1}^n p_i^\downarrow = \sum_{i=1}^n q_i^\downarrow,
\]
where \(\lambda(\rho) \equiv \{p_i^\downarrow\} \in \mathbb{R}^n\) and \(\lambda(\sigma) \equiv \{q_i^\downarrow\} \in \mathbb{R}^n\) are the spectra of \(\rho\) and \(\sigma\), respectively, arranged in nonincreasing order \((p_1^\downarrow \geq p_2^\downarrow \geq \cdots \geq p_n^\downarrow)\), \((q_1^\downarrow \geq q_2^\downarrow \geq \cdots \geq q_n^\downarrow)\).

\begin{theorem}\label{the:contiunity}
Let \( \ket{\psi} \) and \( \ket{\phi} \) be $n$-qubit pure states.

(1) \textbf{Continuity}: if their trace distance satisfies \( D(\psi, \phi) := \frac{1}{2}\|\psi - \phi\|_1 \leq  \frac{\epsilon}{2\epsilon_{d-1}}\),
where \( D(\cdot, \cdot) \) denotes the trace distance, then
\begin{equation}\label{eq:contin_1}
\left|M^{(s)}_E(\ket{\psi}) - M^{(s)}_E(\ket{\phi})\right|
\leq \epsilon.
\end{equation}

(2) \textbf{Majorization monotonicity}: if $\lambda(\ket{\psi})\prec \lambda(\ket{\phi})$ for the spectrum of the individual marginals of $\ket{\phi}$ and $\ket{\psi}$, then $M^{(s)}_E(\ket{\psi})\geq M^{(s)}_E(\ket{\phi})$.
\end{theorem}
\begin{proof}
Suppose that the Schmidt decomposition of $\ket{\psi}$ and $\ket{\phi}$ in the energy basis are denoted by
\begin{eqnarray}
    \label{eq:SCH-DEC1}
    &&\ket{\psi}=\sum_{i=1}^{d-1}\sqrt{\lambda_i}\ket{\alpha_i^X}\ket{\beta_i^{X^c}},\nonumber\\
    &&\ket{\phi}=\sum_{i=1}^{d-1}\sqrt{\eta_i}\ket{a_i^X}\ket{b_i^{X^c}},\end{eqnarray}
where $\lambda_i$ and $\eta_i$ have been chosen in nonincreasing order.
Then, the Schmidt decomposition gives the same specutrum for the marginals, which can be written in the passive form in the energy basis as following
\begin{eqnarray}
    \label{eq:SCH-DEC2}
    &&\rho^p_X(\ket{\psi})=\rho^p_{X^c}(\ket{\psi})=\sum_{j=0}^{d-1}\lambda_j\ket{j}\bra{j},\nonumber\\
    &&\rho^p_X(\ket{\phi})=\rho^p_{X^c}(\ket{\phi})=\sum_{j=0}^{d-1}\eta_j\ket{j}\bra{j},\end{eqnarray}
where the reduced systems $X$ and $X^c$ are governed by Hamiltonians $H_X=\sum_{j=0}^{d-1}\epsilon_j^X\ket{j}\bra{j}$ and $H_{X^c}=\sum_{j=0}^{d-1}\epsilon_j^{X^c}\ket{j}\bra{j}$.
\\

(1)\  \textbf{Continuity bound:}
\begin{eqnarray}
&&|M^{(s)}_E(\ket{\psi}) - M^{(s)}_E(\ket{\phi})| \nonumber\\
&=&\nonumber \frac{1}{2^{|s|}}|\sum_{X\in\mathcal{P}(s)}\Delta_{X|X^c}(\ket{\psi})-\sum_{X\in\mathcal{P}(s)}\Delta_{X|X^c}(\ket{\phi})|
\\\nonumber &=&
\frac{1}{2^{|s|}}|\sum_{X\in\mathcal{P}(s)}(\Delta_{X|X^c}(\ket{\psi})-\Delta_{X|X^c}(\ket{\phi}))|
\\\nonumber &\overset{(a)}{\leq}&
\frac{1}{2^{|s|}}\sum_{X\in\mathcal{P}(s)}|\Delta_{X|X^c}(\ket{\psi})-\Delta_{X|X^c}(\ket{\phi})| \\ \nonumber
&\overset{(b)}{=}&\frac{1}{2^{|s|}}\sum_{X\in\mathcal{P}(s)}|\sum_{j=0}^{d-1}(\lambda_j-\eta_j)\epsilon_j| \\\nonumber
 &\overset{(c)}{\leq}&
\sum_{j=0}^{d-1} |\lambda_j-\eta_j|\epsilon_{d-1}\\ &\overset{(d)}{\leq}&
\frac{\epsilon}{2\epsilon_{d-1}}\epsilon_{d-1} \nonumber\\
&\overset{}{\leq}&
\epsilon,
\end{eqnarray}
where the validity of (a) follows from the triangle inequality; (b) holds follows from \(\Delta_{X|X^c}(\ket{\psi})-\Delta_{X|X^c}(\ket{\phi})=\sum_{j=0}^{d-1}(\lambda_j-\eta_j)\epsilon_j\) and \(\epsilon_j=\epsilon_j^X+\epsilon_j^{X^c}\) holds~\cite{alimuddin2019bound}; (c) holds follows from the triangle inequality and $\epsilon_j\leq\epsilon_{j+1}$; (d) holds by the inequality $D(\psi, \phi)= \frac{1}{2}\|\psi - \phi\|_1=\frac{1}{2}\sum_{j=0}^{d-1}|\lambda_j-\eta_j|\leq  \frac{\epsilon}{2\epsilon_{d-1}}$.

(2)\ \textbf{Majorization monotonicity}:
\begin{eqnarray}
&&M^{(s)}_E(\ket{\psi}) - M^{(s)}_E(\ket{\phi}) \nonumber\\
&=&\nonumber
\frac{1}{2^{|s|}}\sum_{X\in\mathcal{P}(s)}(\Delta_{X|X^c}(\ket{\psi})-\Delta_{X|X^c}(\ket{\phi}))
\nonumber\\
&\overset{(e)}{=}&\frac{1}{2^{|s|}}\sum_{X\in\mathcal{P}(s)}\sum_{j=0}^{d-1}(\lambda_j-\eta_j)\epsilon_j \\\nonumber
&\overset{(f)}{=}&\sum_{k=0}^{d-2}(\epsilon_{k+1}-\epsilon_{k})\sum_{j=0}^{k}(\eta_j-\lambda_j) \\\nonumber
&\overset{(g)}{\geq}&
0,
\end{eqnarray}
where (e) holds follows from \(\Delta_{X|X^c}(\ket{\psi})-\Delta_{X|X^c}(\ket{\phi})=\sum_{j=0}^{d-1}(\lambda_j-\eta_j)\epsilon_j\) and \(\epsilon_j=\epsilon_j^X+\epsilon_j^{X^c}\) holds~\cite{alimuddin2019bound}; (f) is obtained by rearranging the equality (e) and follows from $\epsilon_k\leq\epsilon_{k+1}$ for any \(k\in\{0,1,\cdots,d-2\}\); (g) holds due to the majorization condition $\lambda(\ket{\psi})\prec \lambda(\ket{\phi})$, i.e., $\sum_{j=0}^{k} \lambda_j\leq \sum_{j=0}^{k}\eta_j $ for any $k\geq 0$.
\end{proof}
\begin{proposition}
  \label{th:monogamy}
Let $s = [1]$ and given the local Hamiltonian $H_{A_i}=\ket{1}\bra{1}$ for every system, then
the following monogamy relation holds:
\begin{eqnarray}
  &&(M_{E}^{(s)}(\ket{\psi}_{A_1\cdots A_i\cdots A_n}))^\alpha\nonumber\\
  &\geq& (M_{E}^{(s)}(\rho_{A_iA_1}))^\alpha+\cdots+(M_{E}^{(s)}(\rho_{A_iA_{i-1}}))^\alpha \nonumber\\
   && +(M_{E}^{(s)}(\rho_{A_iA_{i+1}}))^\alpha+\cdots+(M_{E}^{(s)}(\rho_{A_iA_{n}}))^\alpha,
\end{eqnarray}
where $\rho_{A_iA_j}=\tr_{A_1\cdots A_{i-1}A_{i+1}\cdots A_{j-1}A_{j+1}\cdots A_n}(\ket{\psi}\bra{\psi})$ and $\alpha\geq 1$.
\end{proposition}
\begin{proof}
Due to $s = [1]$, let us consider the bi-separable $A_1|A_2\cdots A_n$ as an example.
Therefore,
$$M_{E}^{(s)}(\ket{\psi}_{A_1\cdots A_i\cdots A_n})=\frac{1}{2}\Delta_{A_1|A_2\cdots A_n}(\ket{\psi}).$$
Since $\Delta^\alpha_{A_1|A_2\cdots A_n}(\ket{\psi})\geq \Delta^\alpha_{A_1|A_2}(\ket{\psi})+\cdots +\Delta^\alpha_{A_1|A_n}(\ket{\psi})$ \cite{sun2024genuine}, we have that
\begin{eqnarray}
   && (M_{E}^{(s)}(\ket{\psi}_{A_1\cdots A_i\cdots A_n}))^\alpha\nonumber\\
    &=&
    \frac{1}{2^\alpha}\Delta^\alpha_{A_1|A_2\cdots A_n}(\ket{\psi})\nonumber \\
    &\geq&
    \frac{1}{2^\alpha}(\Delta^\alpha_{A_1|A_2}(\ket{\psi})+\cdots +\Delta^\alpha_{A_1|A_n}(\ket{\psi}))\nonumber \\
    &=&
    (M_{E}^{(s)}(\rho_{A_1A_2}))^\alpha+\cdots+(M_{E}^{(s)}(\rho_{A_1A_n}))^\alpha.
\end{eqnarray}
\end{proof}
\begin{proposition}\label{the:CE-TRI}
Suppose that $s=[3]$ and a three-qubit state $\ket{\psi}_{ABC}$ governed by the local Hamiltonian $H_i=\ket{1}\bra{1}$ for each system, where the reduced density matrices $\rho_{AB}=\tr_C(\ket{\psi}\bra{\psi})$, $\rho_{AC}=\tr_B(\ket{\psi}\bra{\psi})$ and $\rho_{BC}=\tr_A(\ket{\psi}\bra{\psi})$ have eigenvalues arranged in nonincreasing order denoted by
$\{x_i\}^3_{i=0}$, $\{y_i\}^3_{i=0}$ and $\{z_i\}^3_{i=0}$,
respectively. The state is multipartite entangled if
\begin{equation}
    M^{(s)}_E(\ket{\psi}_{ABC})> \min \{\frac{1}{2}(x_1+x_2),\frac{1}{2}(y_1+y_2), \frac{1}{2}(z_1+z_2), \frac{1}{4}\}.
\end{equation}
Furthermore, a three-qubit state $\ket{\psi}_{ABC}$ is genuinely multipartite entangled if
\begin{equation}
    M^{(s)}_E(\ket{\psi}_{ABC})>\frac{1}{2}.
\end{equation}
\end{proposition}
\begin{proof}
    It is easy to see that $M^{(s)}_E(\ket{\psi})=0$ for a product state $\ket{\psi}$. Therefore, for a tripartite pure state $\ket{\psi}_{ABC}$, the spectra of bipartite reductions satisfy $\lambda(\rho_A)=\lambda(\rho_{BC})$ (similarly for other cuts). If we consider the equal marginal Hamiltonian $H=\sum_ii\ket{i}\bra{i}$ for every party, then we have that~\cite{alimuddin2020independence}
    \begin{equation}
        \label{eq:detal}
        \Delta_{A|BC}(\ket{\psi})\leq \Delta_{A|B}(\rho_{AB})+\Delta_{A|C}(\rho_{AC}),
    \end{equation}
where the equality would hold for the $\mathbb{C}^2\otimes\mathbb{C}^2\otimes\mathbb{C}^2$ systems only.

   By Eq.(\ref{eq:detal}), the following chain of equalities hold:
\begin{eqnarray}\label{eq:genuninelymultipartite}
   && M^{(s)}_E(\ket{\psi}_{ABC}) \nonumber\\
    &=& \nonumber \frac{1}{2^3}\sum_{X\in\mathcal{P}(s)}\Delta_{X|X^c}(\ket{\psi})\\ \nonumber
    &=&
    \frac{1}{4}(\Delta_{A|BC}(\ket{\psi})+\Delta_{B|AC}(\ket{\psi})+\Delta_{C|AB}(\ket{\psi}))\\
    &{\leq}&
    \frac{1}{2}(\Delta_{A|B}(\rho_{AB})+\Delta_{A|C}(\rho_{AC})+\Delta_{B|C}(\rho_{BC})),
\end{eqnarray}
where the equality holds when each qubit is governed by the Hamiltonian $H_i=\ket{1}\bra{1}$.

Without loss of generality, we assume
\(\ket{\psi}_{ABC} = \ket{\phi}_{AB} \ket{\phi}_C \). Thus, we have that
$\Delta_{C|AB}(\ket{\psi}_{ABC})=0.$
From the conditions under which Eq.(\ref{eq:detal}) holds, we can obtain that
$$\Delta_{C|AB}(\ket{\psi}_{ABC})=\Delta_{C|A}(\rho_{AC})+\Delta_{C|B}(\rho_{BC})=0,$$
namely, $\Delta_{C|A}(\rho_{AC})=0$ and $\Delta_{C|B}(\rho_{BC})=0$. Furthermore, Eq.({\ref{eq:genuninelymultipartite}}) is changed to
\begin{equation}
    M^{(s)}_E(\ket{\psi}_{ABC})\leq \frac{1}{2}\Delta_{A|B}(\rho_{AB})=\frac{1}{2}M^{(s)}_E(\rho_{AB}).
\end{equation}
When each qubit is governed by the Hamiltonian $H_i=\ket{1}\bra{1}$, then we have that
\begin{equation}
    \label{eq:equal}
    M^{(s)}_E(\ket{\psi}_{ABC})=\frac{1}{2}M^{(s)}_E(\rho_{AB}).
\end{equation}

If $\rho_{AB}$ is a separable two-qubit state and the reduced subsystems are governed by the same Hamiltonian $H_{A/B}=\ket{1}\bra{1}$, then the ergotropy gap
is bounded by $\Delta_{A|B}(\rho_{AB})\leq \min \{x_1+x_2, \frac{1}{2}\}$ \cite{alimuddin2019bound}, where $\frac{1}{2}$ is the maximum ergotropic gap over the whole state space of separable states.
Furthermore, it implies that
\begin{equation}
    \label{eq:esphi1}
   M^{(s)}_E(\ket{\psi}_{ABC})=\frac{1}{2}\Delta_{A|B}(\rho_{AB})\leq  \min \{\frac{1}{2}(x_1+x_2), \frac{1}{4}\}.
\end{equation}
Similarly, we can obtain that
\begin{eqnarray}
    &&M^{(s)}_E(\ket{\psi}_{ABC})\leq  \min \{\frac{1}{2}(y_1+y_2), \frac{1}{4}\},\\
    &&M^{(s)}_E(\ket{\psi}_{ABC})\leq  \min \{\frac{1}{2}(z_1+z_2), \frac{1}{4}\}.
\end{eqnarray}
Thus, we conclude that $\ket{\psi}_{ABC}$
is multipartite
entangled if
\begin{eqnarray}
    M^{(s)}_E(\ket{\psi}_{ABC})&>& \min \{\frac{1}{2}(x_1+x_2),\frac{1}{2}(y_1+y_2),\nonumber\\
    &\qquad& \frac{1}{2}(z_1+z_2), \frac{1}{4}\}.\nonumber
\end{eqnarray}

If $\rho_{AB}$ is a two-qubit entangled state, then $M^{(s)}_E(\rho_{AB})\leq M^{(s)}_E(\ket{\beta}_{AB})=1$ \cite{puliyil2022thermodynamic}, where $\ket{\beta}_{AB}=\frac{1}{\sqrt{2}}(\ket{00}+\ket{11})$. Therefore, the maximal value of $M^{(s)}_E(\rho_{AB})$ is 1, and it implies that
\begin{equation}
    \label{eq:esphi2}
   M^{(s)}_E(\ket{\psi}_{ABC})=\frac{1}{2}M^{(s)}_E(\rho_{AB})\leq  \frac{1}{2}.
\end{equation}
Thus, we conclude that $\ket{\psi}_{ABC}$
is genuinely multipartite
entangled if $M^{(s)}_E(\ket{\psi}_{ABC})>\frac{1}{2}$.
\end{proof}

\textbf{Remark:} we have seen that \(M^{(s)}_E(\ket{\psi}_{ABC})>\min \{\frac{1}{2}(x_1+x_2),\frac{1}{2}(y_1+y_2),\frac{1}{2}(z_1+z_2),\frac{1}{4}\}\) detects entanglement, where \(\frac{1}{2}(x_1+x_2)\), \(\frac{1}{2}(y_1+y_2)\) and \(\frac{1}{2}(z_1+z_2)\) are the spectral-dependent criterion, and \(\frac{1}{4}\) (the maximal value over all separable two-qubit states) is the dimension-dependent criterion. For instance, considering a three qubit state $\ket{\psi}_{ABC}=\lambda_0\ket{000}+\lambda_3\ket{110}+\lambda_4\ket{111}$ and each having the  Hamiltonian $H_i=\ket{1}\bra{1}$, we have that $$M_{E}^{(s)}(\ket{\psi}_{ABC})=\frac{1}{4}(3-2\sqrt{1-4\lambda_0^2(1-\lambda_0^2)}-\sqrt{1-4\lambda_0^2\lambda_4^2})$$ with $\lambda_0^2+\lambda_3^2+\lambda_4^2=1$. As shown in Fig.\ref{fig:TRIGEN}, when $\lambda_0\in(0.57,0.82)$, there exists the  genuinely multipartite
entangled state $\ket{\psi}_{ABC}$.

\begin{figure}[htpb]
\centering
    \includegraphics[width=1.05\linewidth]{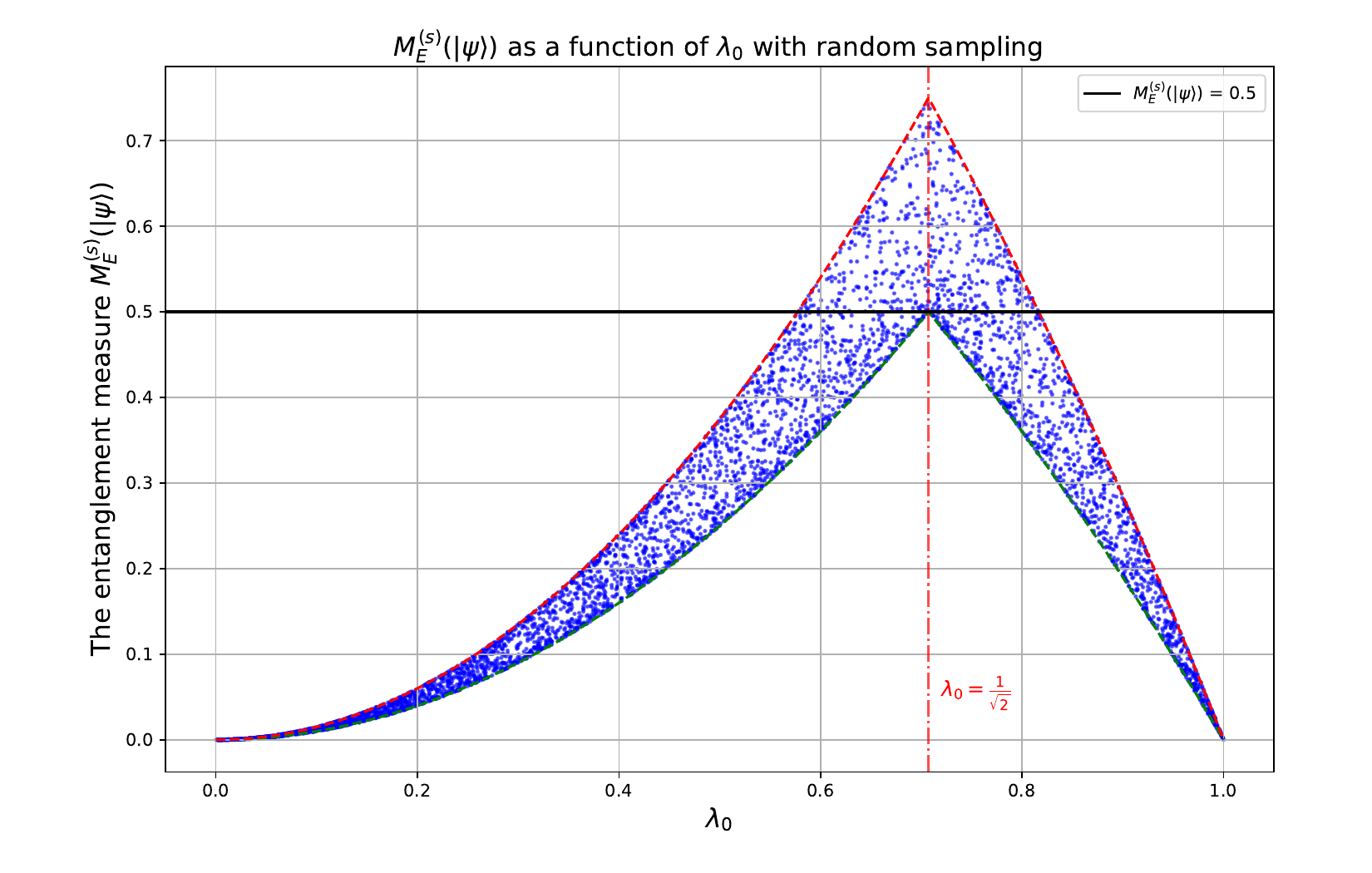}
     \caption{$M^{(s)}_{E}(\ket{\psi}_{ABC})$ as a function of $\lambda_0$ with random sampling and $s=[3]$.}
    \label{fig:TRIGEN}
\end{figure}
Furthermore, we will study bewteen ergotropic-gap concentratable entanglement  with the other measures. The
following entanglement measures \cite{puliyil2022thermodynamic} will be considered in this paper.
\begin{itemize}
    \item Minimum bi-separable erogotropic gap:$$\Delta^{[2]}_{\min}(\ket{\psi})=\min_{X\in \mathcal{P}(s)} \{\Delta_{X|X^c}(\ket{\psi})\}.$$
    \item Bi-separable average erogotropic gap:
    $$\Delta^{[2]}_{\rm avg}(\ket{\psi})=\frac{1}{2^n-2}\sum_{X\in \mathcal{P}(s)}\Delta_{X|X^c}(\ket{\psi}).$$
    \item 2-erogotropic volume:
    $$\Delta^{[2]}_{\rm V}(\ket{\psi})=\Big(\prod_{X\in \mathcal{P}(s)}\Delta_{X|X^c}(\ket{\psi})\Big)^{\frac{1}{2^n-2}}.$$
    \item Fully separable ergotropic gap:
    $$\Delta_{A_1|\cdots|A_n}^{(n)}(\ket{\psi})= \sum_{i=1}^n \tr(\rho_{A_i}^p H_{A_i}).$$
\end{itemize}

\begin{proposition}
 \label{the:relatinship}
 For any $n$-qubit pure state $\ket{\psi}$ and $s=[n]$:

(1) The ergotropic-gap concentratable entanglement $M^{(s)}_E(\ket{\psi})$ satisfies
\begin{equation}
 \Delta^{[2]}_{\min}(\ket{\psi})\leq M^{(s)}_E(\ket{\psi})\leq \Delta^{[2]}_{\rm avg}(\ket{\psi}),
\end{equation}
where $\lim_{n\to \infty}  M^{(s)}_E(\ket{\psi})= \Delta^{[2]}_{\rm avg}(\ket{\psi})$.

(2) Under the local Hamiltonian $H_i=\ket{1}\bra{1}$ for every system, the ergotropic-gap concentratable entanglement $M^{(s)}_E(\ket{\psi})$ satisfies
\begin{equation}
\frac{1}{2^{2^n-1+n}}(\Delta^{[2]}_{\rm V}(\ket{\psi}))^{2^n-2}\leq M^{(s)}_E(\ket{\psi}) \leq \Delta_{A_1|\cdots|A_n}^{(n)}(\ket{\psi}).
\end{equation}
\end{proposition}
\begin{proof}
(1) For $M^{(s)}_E(\ket{\psi})$, we have that
\begin{eqnarray}
 M^{(s)}_E(\ket{\psi})&=&\frac{1}{2^{|s|}}\sum_{X\in \mathcal{P}(s)}\Delta_{X|X^c}(\ket{\psi}) \nonumber\\
 &\overset{(a)}{\geq}& \frac{1}{2^{|s|}}\sum_{X\in \mathcal{P}(s)}\Big(\min_{X\in \mathcal{P}(s)}\{\Delta_{X|X^c}(\ket{\psi})\}\Big)\nonumber\\
 &=& \frac{1}{2^{|s|}}\sum_{X\in \mathcal{P}(s)}\Delta^{[2]}_{\min}(\ket{\psi})\nonumber\\
 &=& \Delta^{[2]}_{\min}(\ket{\psi}),
\end{eqnarray}
where the inequality (a) is due to  $\Delta_{X|X^c}(\ket{\psi})\geq \min_{X\in \mathcal{P}(s)} \{\Delta_{X|X^c}(\ket{\psi})\}$.

\begin{eqnarray}
    M^{(s)}_E(\ket{\psi})&=&\frac{1}{2^{|s|}}\sum_{X\in \mathcal{P}(s)}\Delta_{X|X^c}(\ket{\psi}) \nonumber\\
    &=&\frac{2^n-2}{2^{|s|}}\Delta^{[2]}_{\rm avg}(\ket{\psi}) \nonumber\\
    &=&(1-\frac{1}{2^{|s|-1}})\Delta^{[2]}_{\rm avg}(\ket{\psi}) \nonumber\\
    &\overset{(b)}{\leq}& \Delta^{[2]}_{\rm avg}(\ket{\psi}),
\end{eqnarray}
where (b) holds because $s=[n]$ implies $2^n=2^{|s|}$.  Moreover, as $n$ approaches infinity, $M^{(s)}_E(\ket{\psi})= \Delta^{[2]}_{\rm avg}(\ket{\psi})$.

(2) Since $\Delta^{[2]}_{\rm V}(\ket{\psi})=\Big(\prod_{X\in \mathcal{P}(s)}\Delta_{X|X^c}(\ket{\psi})\Big)^{\frac{1}{2^n-2}}$, thus we have that
\begin{eqnarray}
    && \log_2(\Delta^{[2]}_{\rm V}(\ket{\psi}))\nonumber\\
    &=&\frac{1}{2^n-2}\sum_{X\in\mathcal{P}(s)}\log_2(\Delta_{X|X^c}(\ket{\psi})\nonumber\\
    &\overset{(c)}{=}&\frac{1}{2^n-2}\sum_{X\in\mathcal{P}(s)}\log_2(2\lambda^X_{\min})\nonumber\\
    &=&\frac{1}{2^n-2}(2^n+\sum_{X\in\mathcal{P}(s)}\log_2(\lambda^X_{\min}))\nonumber\nonumber\\
    &\overset{(d)}{\leq}&\frac{1}{2^n-2}(2^n+\sum_{X\in \mathcal{P}(s)}p_X\log_2(\lambda^X_{\min}))\nonumber\\
    &\overset{(e)}{\leq}&\frac{1}{2^n-2}(2^n+\log_2(\sum_{X\in \mathcal{P}(s)}p_X\lambda^X_{\min}))\nonumber\\
    &\overset{(f)}{\leq}&\frac{1}{2^n-2}(2^n+\log_2(\sum_{X\in \mathcal{P}(s)}\lambda^X_{\min}))\nonumber\\
    &=&\frac{1}{2^n-2}(2^n-1+\log_2(\sum_{X\in \mathcal{P}(s)}\Delta_{X|X^c}(\ket{\psi}))\nonumber\\
    &=&\frac{1}{2^n-2}\log_2(2^{2^n-1+n}M^{(s)}_E(\ket{\psi})),
\end{eqnarray}
where the validity of $(c)$ follows from $\Delta_{X|X^c}(\ket{\psi})=2\lambda^X_{\min}$ for the given local Hamiltonian $H_i=\ket{1}\bra{1}$ for every system, and $\lambda^X_{\min} (0 <\lambda^X_{\min}< 1)$ is the minimal eigenvalue of the quantum state $\rho_X=\tr_{X^c}(\ket{\psi}\bra{\psi})$; $(d)$ holds follows that $\log_2(\lambda^X_{\min})<0$ and $p_X$ is a probability value and $\sum_{X\in \mathcal{P}(s)}p_X=1$; the correctness of $(e)$ comes from the concavity of the logarithmic function $\log_2(\cdot)$; $(f)$ holds follows from $p_X\lambda^X_{\min}\leq \lambda^X_{\min}$.

Therefore, we can obtain that
\begin{equation}
    M^{(s)}_E(\ket{\psi})\geq \frac{1}{2^{2^n-1+n}}(\Delta^{[2]}_{\rm V}(\ket{\psi}))^{2^n-2}.
\end{equation}

Since $\Delta_{X|X^c}(\ket{\psi})\leq \Delta_{A_1|\cdots|A_n}^{(n)}(\ket{\psi})$ \cite{puliyil2022thermodynamic}, then we have that
$$ M^{(s)}_E(\ket{\psi}) \leq \Delta_{A_1|\cdots|A_n}^{(n)}(\ket{\psi}).$$
\end{proof}

\textbf{Remark:} Fully separable ergotropic gap cannot capture genuine entanglement as it can take non-zero value for a non genuine entangled states. For instance, considering a four qubit system and each having the  Hamiltonian $H_i=\ket{1}\bra{1}$ and $s=[4]$, we give the following quantum states:   $\ket{\Psi}_{ABCD}=\cos{\theta}\ket{0000}+\sin{\theta}\ket{1111}$ and $\ket{\Phi}_{ABCD}=\cos^2{\theta}\ket{0000}+\cos{\theta}\sin{\theta}\ket{0011}+\cos{\theta}\sin{\theta}\ket{1100}+\sin^2{\theta}\ket{1111}$. Thus,
\begin{eqnarray}
  \Delta_{A|B|C|D}(\ket{\Psi}) &=& \Delta_{A|B|C|D}(\ket{\Phi}) \nonumber\\
   &=& \begin{cases}
\displaystyle 4\sin^2{\theta}, & \text{if } 0\leq\theta\leq \frac{\pi}{4}, \\
\displaystyle 4\cos^2{\theta}, & \text{if } \frac{\pi}{4}<\theta\leq \frac{\pi}{2}.
\end{cases}
\end{eqnarray}
However, we can obtain that
\begin{equation}
M^{(s)}_E(\ket{\Psi})=
\begin{cases}
\displaystyle \frac{7}{4}\sin^2{\theta}, & \text{if } 0\leq\theta\leq \frac{\pi}{4}, \\
\\
\displaystyle \frac{7}{4}\cos^2{\theta}, & \text{if } \frac{\pi}{4}<\theta\leq \frac{\pi}{2}.
\end{cases}
\end{equation}
\begin{equation}
M^{(s)}_E(\ket{\Phi})=
\begin{cases}
\displaystyle \sin^2{\theta}+\frac{1}{2}\sin^4{\theta}, & \text{if } 0\leq\theta\leq \frac{\pi}{4}, \\
\\
\displaystyle \cos^2{\theta}+\frac{1}{2}\cos^4{\theta}, & \text{if } \frac{\pi}{4}<\theta\leq \frac{\pi}{2}.
\end{cases}
\end{equation}
Therefore, these two states can be completely distinguished by ergotropic-gap concentratable entanglement $M^{(s)}_E$, as shown in Fig.\ref{fig:FullyeVsEG}.
\begin{figure}[ht]
\centering
    \includegraphics[width=1.05\linewidth]{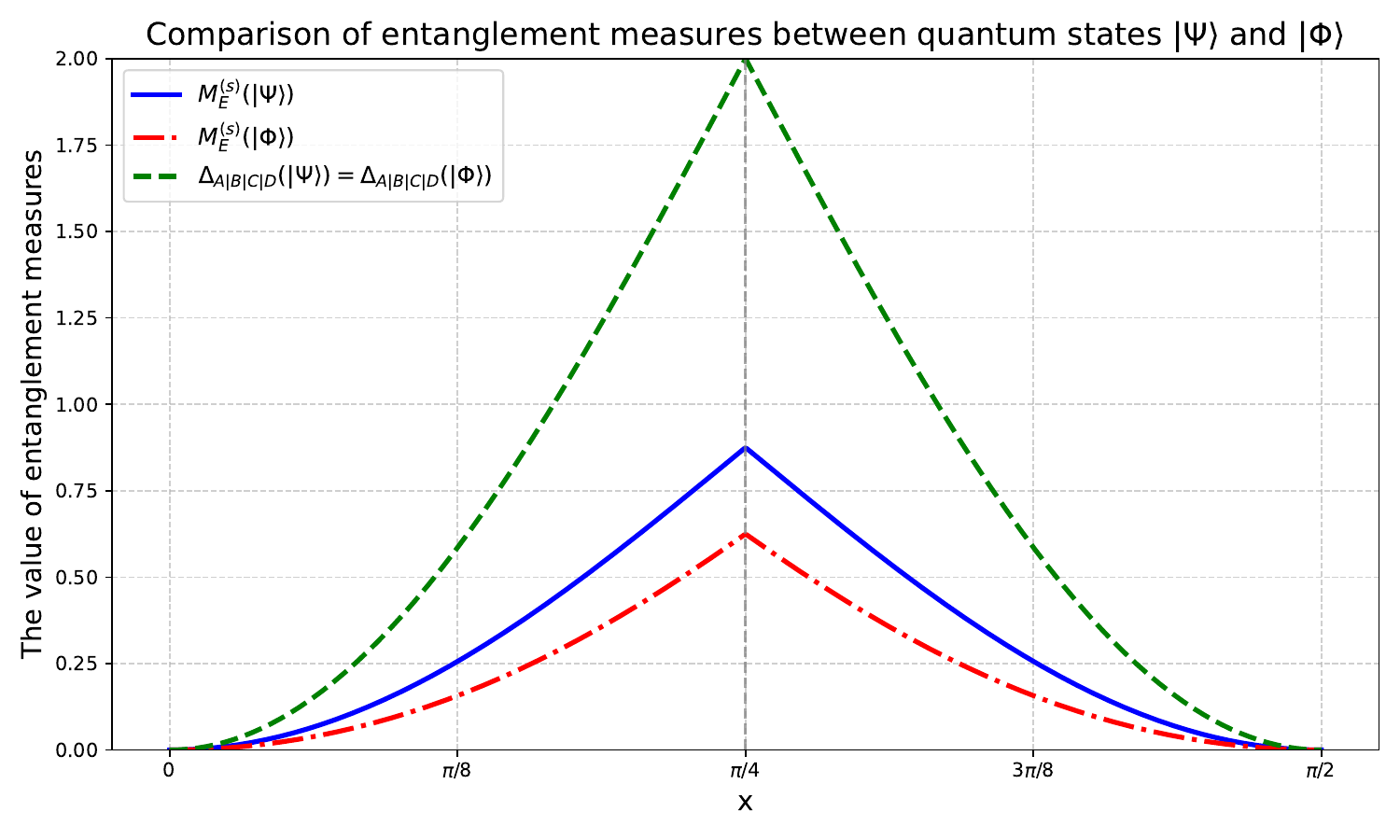}
     \caption{Comparison of entanglement measures between quantum states $\ket{\Psi}$ and $\ket{\Phi}$.}
    \label{fig:FullyeVsEG}
\end{figure}
\section{Example} \label{sec:example}
In this section, we present several examples of ergotropic-gap concentratable entanglement for the most general three qubit pure state. Furthermore, we consider the entanglement for
\(n\)-qubit GHZ and W states. Furthermore, we then investigate the ergotropic-gap  concentratable entanglement in a four-partite star quantum network.

\subsection{The simplest multipartite system $(\mathbb{C}^2)^{\otimes3}$}
Let us look at the simplest multipartite system---the three-qubit system---with ground and excited state energies, 0 and 1 respectively for all the parties. The most general three qubit pure state in the generalised Schmidt form is given by \cite{acin2000generalized},
\begin{equation}
    \label{eq:exthreequbit}
\ket{\psi}_{ABC}=\lambda_0\ket{000}+\lambda_1e^{\mathbf{i}\varphi}\ket{100}+\lambda_2\ket{101}+\lambda_3\ket{110}+\lambda_4\ket{111},
\end{equation}
where $\lambda_i\geq0$, $\sum_i\lambda_i^2=1$ and $0\leq \varphi\leq \pi$.

By definition, the spectrum of each single marginal will be the same with the remaining two-party marginal. Since the Hamiltonian for the marginal systems are $H_A = H_B = H_C = \ket{1}\bra{1}$, passive state energy will be equal to the smallest eigenvalue. In terms of the generalised Schmidt coefficients, the passive state energies are denoted by
   \begin{eqnarray*}
   && \tr(\rho_A^pH_A)=\frac{1}{2}\Big(1-\sqrt{1-4\lambda_0^2(1-\lambda_0^2-\lambda_1^2)}\Big)=\frac{\Delta_{A|BC}}{2},\\
    && \tr(\rho_B^pH_B)=\frac{1}{2}\Big(1-\sqrt{1-4(\lambda_0^2(\lambda_3^2+\lambda_4^2)+\alpha)}\Big)=\frac{\Delta_{B|AC}}{2},\\
    && \tr(\rho_C^pH_C)=\frac{1}{2}\Big(1-\sqrt{1-4(\lambda_0^2(\lambda_2^2+\lambda_4^2)+\alpha)}\Big)=\frac{\Delta_{C|AB}}{2},
\end{eqnarray*}
where $\alpha=|\lambda_1\lambda_4e^{\mathbf{i}\varphi}-\lambda_2\lambda_3|^2$ and $\Delta_{i|jk}$ is the bi-separable ergotropic gap.

The explicit expressions of the bi-separable ergotropic gaps for a generic three qubit state (\ref{eq:exthreequbit}) are given. Setting different coefficients equal to zero, one can get different classes of three qubit states such as generalised GHZ states, tri-Bell states and extended GHZ etc \cite{acin2000generalized}. With those conditions, in Table \ref{tab:eg_label} we analyze the ergotropic-gap concentratable entanglement
 for different classes of states.
 \begin{table*}[htbp]
  \centering
    \caption{The ergotropic-gap concentratable entanglement of three qubit pure states with two or less parameters}
    \begin{tabular}{c|c|c}
        \hline
        Type of State & Constraint & entanglement $M^{(s)}_E$  \\
        \hline
        Generalised GHZ & $\lambda^2_0\leq0.5$ & $1.5\lambda^2_0$ \\
        ($\lambda_0\lambda_4\neq 0=\lambda_i, i\in\{1,2,3\}$) & $\lambda^2_0\geq0.5$ & $1.5(1-\lambda^2_0)$ \\
        \hline
         & $\lambda^2_0\geq0.5$ & $1-\lambda_0^2$ \\
        Tri-Bell & $\lambda^2_3\geq0.5$ & $1-\lambda_3^2$ \\
        ($\lambda_4=\lambda_1=0$) & $\lambda^2_2\geq0.5$ & $1-\lambda_2^2$ \\
          & $\lambda^2_0,\lambda^2_2,\lambda^2_3\leq0.5$ & 0.5 \\
        \hline
         & $\lambda_1\neq0,\lambda_4^2\leq 0.5$ & $0.25(1-\sqrt{1-4\lambda^2_0\lambda_4^2}+4\lambda_4^2)$ \\
         & $\lambda_1\neq0,\lambda_4^2\geq 0.5$ & $0.25(1-\sqrt{1-4\lambda^2_0\lambda_4^2}+4(1-\lambda_4^2))$ \\
        Extended GHZ & $\lambda_2\neq0,\lambda_0^2\leq 0.5$ & $0.25(1-\sqrt{1-4\lambda^2_0\lambda_4^2}+4\lambda_0^2)$ \\
($\lambda_i=\lambda_j=0\neq\lambda_k,i,j,k\in\{1,2,3\}$) & $\lambda_2\neq0,\lambda_0^2\geq 0.5$ & $0.25(1-\sqrt{1-4\lambda^2_0\lambda_4^2}+4(1-\lambda_0^2))$  \\
         & $\lambda_3\neq0,\lambda_0^2\leq 0.5$ & $0.25(1-\sqrt{1-4\lambda^2_0\lambda_4^2}+4\lambda_0^2)$ \\
         & $\lambda_3\neq0,\lambda_0^2\geq 0.5$& $0.25(1-\sqrt{1-4\lambda^2_0\lambda_4^2}+4(1-\lambda_0^2))$  \\
        \hline
    \end{tabular}
    \label{tab:eg_label}
\end{table*}

Furthermore, we consider the inequivalence among different entanglement measures. Finding inequivalent monotones is quite important as it helps us to comment on the state inter-convertibility of the different multipartite entangled states. As the entanglement measures are LOCC monotone,  $E(\ket{\psi})>E(\ket{\phi})$ for an entangled pair $\ket{\psi}$ and $\ket{\phi}$ implies abandoned transformation $\ket{\phi}\rightarrow\ket{\psi}$ under any LOCC. Two measures $E$ and $E'$ are called inequivalent if there exist states $\ket{\alpha}$ and $\ket{\beta}$ such that $E(\ket{\alpha})>E(\ket{\beta})$ whereas $E'(\ket{\alpha})<E'(\ket{\beta})$ which immediately says neither one can be converted to each other \cite{puliyil2022thermodynamic}. We have compared our proposed measure with minimum bi-separable erogotropic gap $\Delta_{\min}^{[2]}$, bi-separable average erogotropic gap $\Delta_{\rm avg}^{[2]}$, 2-erogotropic volume $\Delta_{\rm V}^{[2]}$ and ergotropic fill $\Delta^{[2]}_F$ \cite{puliyil2022thermodynamic}.
\begin{table*}[htbp]
\centering
\caption{ In equivalence among the different entanglement measures.}
\label{tab:equivalence}
\begin{tabular}{l c c c c c}
\hline
State & $M_{E}^{(s)}$&$\Delta_{\min}^{{[2]}}$ & $\Delta_{\rm avg}^{[2]}$ & $\Delta_{\rm F}^{[2]}$  & $\Delta_{\rm V}^{[2]}$ \\
\hline
$|\psi\rangle_{ABC} = \frac{1}{\sqrt{3}}|000\rangle + \sqrt{\frac{2}{3}}|111\rangle$& 0.5
    & 0.667 & 0.667 & 0.667  & 0.667 \\
$|\phi\rangle_{ABC} = \frac{1}{\sqrt{2}}|000\rangle + \sqrt{\frac{9}{32}}|110\rangle + \sqrt{\frac{7}{32}}|111\rangle$ & 0.562
    & 0.250 & 0.750 & 0.559  & {0.630} \\ 
$|\chi\rangle_{ABC} = \frac{1}{\sqrt{8}}|000\rangle + \sqrt{\frac{3}{2}}|111\rangle$&0.375
    & {0.500} & 0.500 & 0.500  & 0.500 \\ 
$|\zeta\rangle_{ABC} = \sqrt{\frac{3}{8}}|000\rangle + \sqrt{\frac{1}{3}}|110\rangle + \sqrt{\frac{7}{24}}|111\rangle$ &0.438
    & 0.250 & 0.583 & 0.479  & 0.520 \\
$|\eta\rangle_{ABC} = \frac{3}{\sqrt{50}}|000\rangle + \sqrt{\frac{41}{50}}|111\rangle$ & 0.270
    & 0.360 & 0.360 & 0.360  & 0.360 \\
\hline
\end{tabular}
\end{table*}

Firstly, we explicitly show that the ergotropic-gap  concentratable entanglement \(M_{E}^{(s)}\) and the bi-separable average erogotropic gap $\Delta_{\rm avg}^{[2]}$ are equivalent due to
$$M^{(s)}_E(\ket{\psi})
    =(1-\frac{1}{2^{|s|-1}})\Delta^{[2]}_{\rm avg}(\ket{\psi}).$$
Thus, they will always maintain the same order for any pair of states while other measures are independent with each other (example of pair of states are shown in Table \ref{tab:equivalence}).

Furthermore, we have that
\begin{itemize}
    \item \{\(M_{E}^{(s)}\), \(\Delta_{\min}^{{[2]}}\)\} show different order for the pair of states $\{\ket{\psi}_{ABC},\ket{\phi}_{ABC}\}$.
    \item \{\(M_{E}^{(s)}\), \(\Delta_{\rm F}^{[2]}\)\} show different order for the pair of states $\{\ket{\psi}_{ABC},\ket{\phi}_{ABC}\}$.
    \item \{\(M_{E}^{(s)}\), \(\Delta_{\rm V}^{[2]}\)\} show different order for the pair of states $\{\ket{\psi}_{ABC},\ket{\phi}_{ABC}\}$.
\end{itemize}
Thus, $M_{E}^{(s)}/\Delta_{\rm avg}^{[2]}$, $\Delta_{\min}^{[2]}$,  $\Delta_{\rm V}^{[2]}$,and $\Delta^{[2]}_{\rm F}$ are inequivalent measures.
\subsection{\texorpdfstring{$n$-qubit GHZ states and W states}{n-qubit GHZ states and W states}}
Under the framework of reversible LOCC, GHZ and W states represent maximally entangled states within their respective equivalence classes. While GHZ states generally exhibit a higher degree of entanglement as quantified by standard pure-state measures, W states are distinguished by their enhanced robustness: entanglement persists within the remaining subsystem even upon the measurement or loss of constituent qubits. In this section, we focus on these two types of entangled states due to their significant applications in quantum computing \cite{horodecki2009quantum}.

Firstly, we will give several illustrative cases of  ergotropic-gap concentratable entanglement \(M_{E}^{(s)}\) for $n$-qubit GHZ states and W states. For \( n \geq 3 \), the \( n \)-qubit GHZ and W states, expressed in the computational basis, are given by
\begin{equation}
\ket{\text{GHZ}_n} = \frac{1}{\sqrt{2}} \left( \ket{0}^{\otimes n} + \ket{1}^{\otimes n} \right),
\end{equation}

\begin{equation}
\ket{W_n} = \frac{1}{\sqrt{n}} \sum_{i=1}^{n} \ket{0 \dots 1_i \dots 0}.
\end{equation}

As shown in Fig.~\ref{fig:GHZvsW-Df}, we plot the difference \(\Delta = M_{E}^{(s)}(\ket{\text{GHZ}_n}) - M_{E}^{(s)}(\ket{W_n})\) for various subsystem sizes \(|s|\).
Meanwhile, for \(s = [n]\), we have
\[M^{(s)}_E(\ket{\text{GHZ}_n}) = 1 - \frac{1}{2^{n-1}}\]
and
\[
M^{(s)}_E(\ket{W}_n) = 1-\frac{1}{2^{n-1}} \binom{n-1}{\lfloor\frac{n-1}{2}\rfloor}.
\]
For \(s = [n-1]\), we have
\[
M^{(s)'}_E(\ket{\text{GHZ}}_n) = 1 - \frac{1}{2^{n-1}}
\]
and
\[
M^{(s)'}_E(\ket{W}_n) = \frac{1}{2}-\frac{1}{2^{n}} \binom{n-1}{\lfloor\frac{n-1}{2}\rfloor},
\]
where $\lfloor x\rfloor$ represents the floor function.

\begin{figure}
    \centering
    \includegraphics[width=1.05\linewidth]{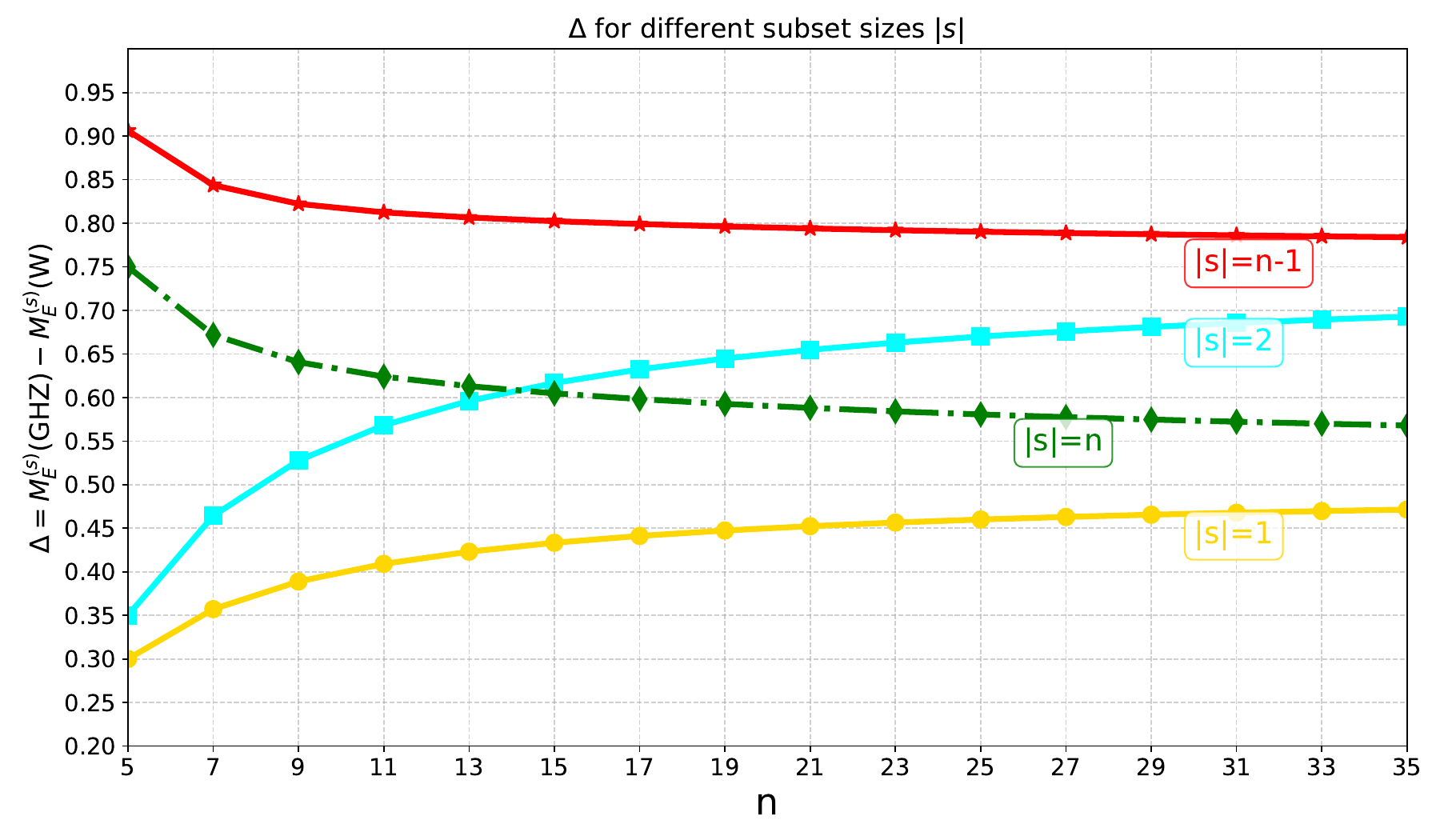}
    \caption{The difference \(\Delta=M^{(s)}_E(\ket{\text{GHZ}}_n) - M^{(s)}_E(\ket{W}_n)\) for various subsystem sizes \(|s|\).}
    \label{fig:GHZvsW-Df}
\end{figure}

\begin{figure}[htbp]
    \centering
    \includegraphics[width=1.05\linewidth]{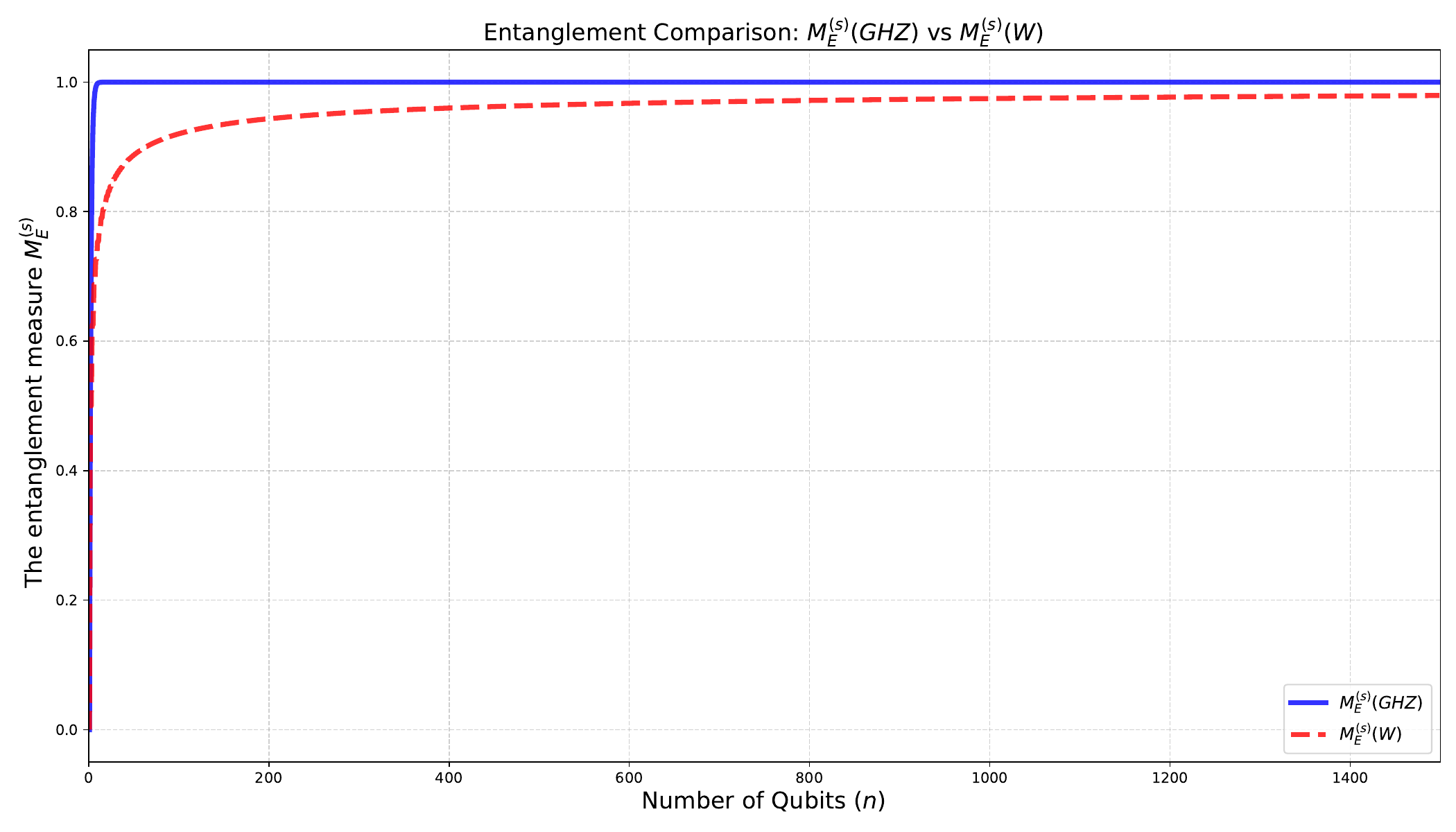}
    \caption{The ergotropic-gap concentratable entanglements serves
as an effective measure for distinguishing GHZ states
from W states.}
    \label{fig:GHZvsWEG}
\end{figure}
As shown in Fig.~\ref{fig:GHZvsWEG}, both \(M^{(s)}_{E}(\ket{\text{GHZ}_n})\) and \(M^{(s)}_E(\ket{W_n})\) increase slowly and gradually approach 1 as \(n\) increases. These results show that for large $n$, the GHZ and W states can be effectively distinguished using small subsystems, i.e., when the subsystem size
\(|s|\) is small. In both cases, we observe that \(M^{(s)}_E(\ket{\text{GHZ}_n}) > M^{(s)}_E(\ket{W_n})\), indicating that \(M^{(s)}_E\) is an effective measure for distinguishing GHZ states from W states.
\subsection{A four-partite star quantum network}\label{sec:starnetwork}
Consider following quantum state
\begin{equation}
    \ket{\psi}=
(\cos{\theta}\ket{00}+\sin{\theta}\ket{11})^{\otimes 3}
\end{equation}
It can be seen that this quantum state is composed of the tensor product of three EPR-like states. When grouping the first qubit from each pair into subsystem $A$, the state $\ket{\psi}$ represents a four-partite pure state in the Hilbert space \(\mathcal{H}_A\otimes\mathcal{H}_B\otimes\mathcal{H}_C\otimes\mathcal{H}_D\)
 with dimension \(8 \otimes 2 \otimes 2 \otimes 2\). Thus, we have
\begin{eqnarray}
\ket{\psi}_{ABCD}
&=&
\cos^3{\theta} \ket{0000}+\cos^2{\theta}\sin{\theta}\ket{1001} \nonumber\\
&&+ \cos^2{\theta}\sin{\theta} \ket{2010} + \cos{\theta}\sin^2{\theta} \ket{3011} \nonumber \\
&& + \cos^2{\theta}\sin{\theta}\ket{4100} + \cos{\theta}\sin^2{\theta} \ket{5101} \nonumber\\
&&+ \cos{\theta}\sin^2{\theta} \ket{6110} +\sin^3{\theta} \ket{7111}.
\end{eqnarray}
 For \(s = [n]\) and the local Hamiltonian $H_A =\ket{7}\bra{7}$ for the first subsystem $A$ and local Hamiltonian $H_X =\ket{1}\bra{1}$ for the subsystem $X\in\{B,C,D\}$, the ergotropic-gap  concentratable entanglement \(M_{E}^{(s)}(\ket{\psi})\) can be expressed as
\begin{eqnarray}
\label{eq:starnetwork}
 M_{E}^{(s)}(\ket{\psi})&=& \frac{1}{8} (\Delta_{A|BCD}(\ket{\psi}) + \Delta_{B|ACD}(\ket{\psi})\nonumber\\
 &&+\Delta_{C|ABD}(\ket{\psi}) +\Delta_{D|ABC}(\ket{\psi})\nonumber\\ &&+\Delta_{AB|CD}(\ket{\psi})+\Delta_{AC|BD}(\ket{\psi})\nonumber\\
 &&+\Delta_{AD|BC}(\ket{\psi})\\
&=& \frac{1}{4}(7\sin^6{\theta}+3\cos^4{\theta}\sin^4{\theta}+6\cos^2{\theta}\sin^4{\theta}).\nonumber
\end{eqnarray}

We conduct a comparative study of entanglement behavior in the four-partite star quantum network using four entanglement measures: minimum bi-separable erogotropic gap $\Delta_{\min}^{[2]}$, bi-separable average erogotropic gap $\Delta_{\rm avg}^{[2]}$, 2-erogotropic volume $\Delta_{\rm V}^{[2]}$, and the ergotropic-gap concentratable entanglement $M_{E}^{(s)}$. As shown in Fig.~\ref{fig:Starnetwork}, $M_{E}^{(s)}$ exhibits superior sensitivity to
\(\theta\)-dependent entanglement variations in  \(\ket{\psi(\theta)}\) compared to $\Delta_{\rm V}^{[2]}$. Specifically,  $M_{E}^{(s)}$ is better than $\Delta_{\min}^{[2]}$ to distinguish quantum states \(\ket{\psi(\theta)}\) with $\theta \notin(1.196,1.946)\cup(4.338,5.087)$. For the system size $n=4$  particles, our proposed entanglement measure $M_{E}^{(s)}$ exhibits lower distinguishing capability than $\Delta_{\rm avg}^{[2]}$. However, as the number of particles increases, our measures approach the performance of $\Delta_{\rm avg}^{[2]}$ in state discrimination.
\begin{figure}[ht]
\centering
    \includegraphics[width=1.05\linewidth]{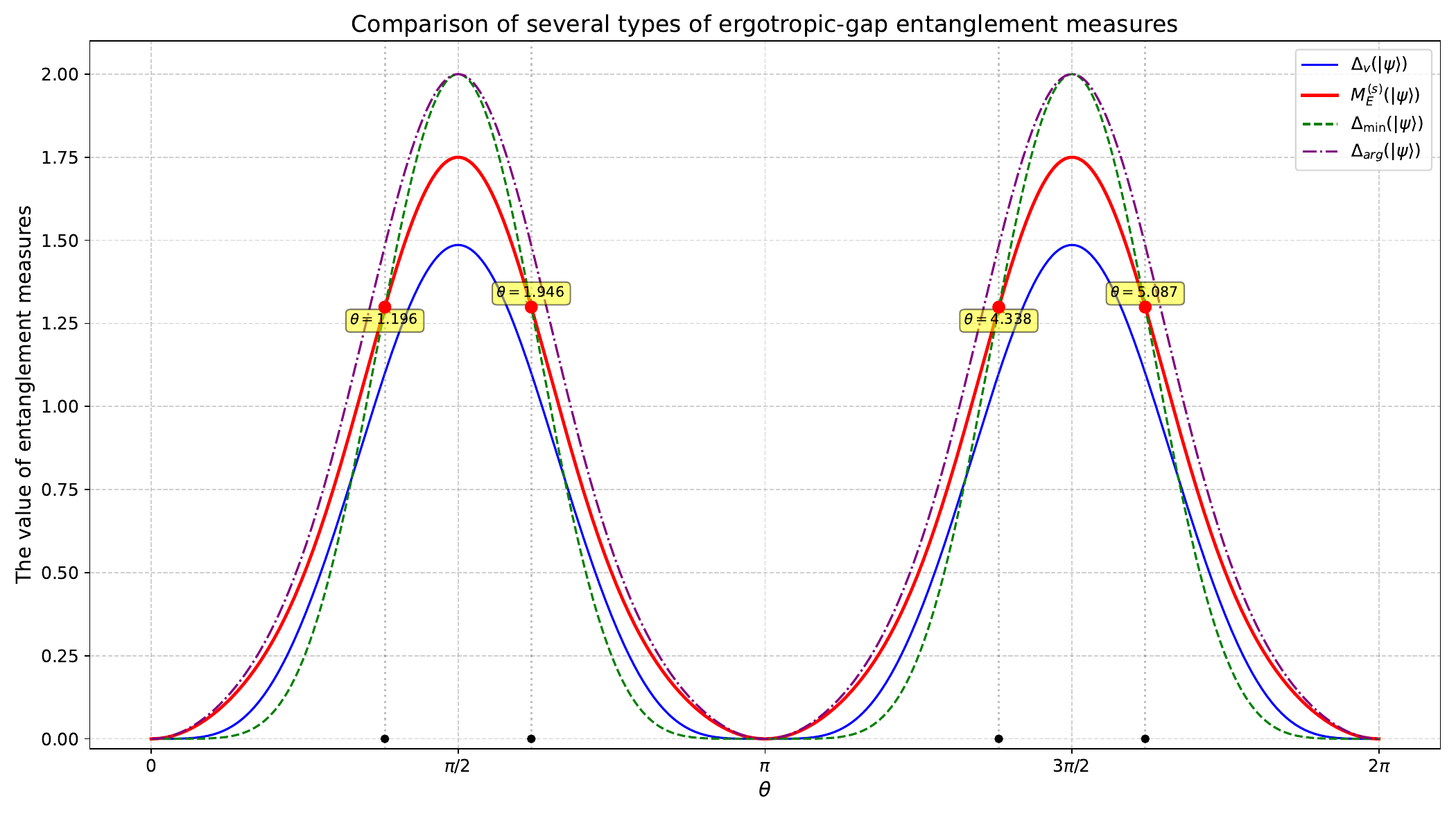}
     \caption{Comparison of  minimum bi-separable erogotropic gap $\Delta_{\min}^{[2]}$, bi-separable average erogotropic gap $\Delta_{\rm avg}^{[2]}$, 2-erogotropic volume $\Delta_{\rm V}^{[2]}$, and the ergotropic-gap concentratable entanglement $M_{E}^{(s)}$.}
    \label{fig:Starnetwork}
\end{figure}
\section{Conclusion}\label{sec:Conclusion}
In this work, we have established a novel thermodynamic framework for quantifying  multipartite entanglement. By leveraging fundamental thermodynamic quantities, we introduced two new families of entanglement measures: the ergotropic-gap concentratable entanglement \(M_E^{(s)}\) and the battery capacity-gap concentratable entanglement \(M_B^{(s)}\). These measures quantify the distribution and concentration of entanglement resources across arbitrary subsets of qubits within a multipartite quantum system. We shown that $M_E^{(s)}$ and $M_B^{(s)}$ are equivalent for systems governed by Hamiltonians featuring equispaced energy levels, demonstrating a fundamental link between these thermodynamic quantities in physically relevant settings. Furthermore, we established that the ergotropic-gap concentratable entanglement $M_E^{(s)}$ constitutes a well-defined multipartite entanglement measure, and we provided rigorous proofs that it satisfies crucial properties expected of a valid entanglement quantifier: non-negativity, monotonicity under LOCC on average, continuity with respect to state perturbations, and majorization monotonicity.

Beyond its axiomatic foundations, we demonstrated the practical utility of the ergotropic-gap concentratable entanglement. We derived a sufficient criterion based on $M_E^{(s)}$ for detecting genuine multipartite entanglement in three-qubit systems. We also explored the monogamy relations satisfied by $M_E^{(s)}$ and established its relationships with other prominent thermodynamic entanglement measures, such as the minimum bi-separable ergotropic gap, the bi-separable average ergotropic gap, the 2-ergotropic volume, and the fully separable ergotropic gap, highlighting its distinct role within the landscape of entanglement quantifiers.

The effectiveness of $M_E^{(s)}$ was concretely demonstrated through its application to paradigmatic multipartite entangled states. Notably, it provides a clear and quantifiable distinction between the fundamentally inequivalent entanglement structures of multi-qubit GHZ states and W states, reflecting their differing resource properties under thermodynamic considerations. Furthermore, we investigated entanglement in a specific four-partite star-shaped quantum network configuration. Intriguingly, our analysis revealed that in certain parameter regimes, the ergotropic-gap concentratable entanglement of the global network state exceeds the value associated with the original bipartite entangled pairs used to construct the network, suggesting a form of entanglement concentration facilitated by the network structure when viewed thermodynamically.

In summary, this work bridges quantum thermodynamics and entanglement theory by introducing operationally meaningful concentratable entanglement measures derived from work extraction capabilities. The ergotropic-gap and battery capacity-gap concentratable entanglements offer versatile tools for analyzing the distribution and quantification of entanglement in complex many-body systems and quantum networks. In future experiments, we hope to detect and characterize multi-body entanglement through thermodynamic methods \cite{niu2024experimental,joshi2024experimental} .
\section*{Acknowledgments}
The authors are very grateful to anonymous reviewers for constructive comments that have greatly helped to improve the quality of this paper. This research was supported by the Nation Natural Science Foundation of China under Grant No.12301590 and Science Research Project of Hebei Education Department under Grant No.BJ2025061.
\bibliographystyle{unsrt}
\bibliography{main}
\fussy
\end{document}